\documentclass[english,nolineno]{socg-lipics-v2021}
\EventEditors{Xavier Goaoc and Michael Kerber}
\EventNoEds{2}
\EventLongTitle{38th International Symposium on Computational Geometry (SoCG 2022)}
\EventShortTitle{SoCG 2022}
\EventAcronym{SoCG}
\EventYear{2022}
\EventDate{June 7--10, 2022}
\EventLocation{Berlin, Germany}
\EventLogo{}
\SeriesVolume{224}
\ArticleNo{XX}  

\usepackage{xspace} 
\usepackage{enumerate}
\usepackage{color}
\usepackage[usenames,dvipsnames,svgnames,table]{xcolor}

\usepackage{thm-restate}

\usepackage{graphicx}
\usepackage{setspace}
\usepackage{mathrsfs}
\usepackage{amssymb,latexsym,amsmath}
\usepackage{amsthm}
\usepackage[T1]{fontenc}
\usepackage{bm}
\usepackage{mathtools}

\usepackage{float}
\newfloat{Algorithm}{t}{lop}

\usepackage{todonotes}

\usepackage[noend]{algpseudocode}   

\newcommand{\ignore}[1]{}

\newcommand{\Start}{\text{S{\scriptsize tart}}}
\newcommand{\End}{\text{E{\scriptsize nd}}}

\newcommand{\gP}{\mathcal{P}}

\newcommand{\gB}{\mathcal{B}}
\newcommand{\gS}{\mathcal{S}}

\newif\ifinmain
\inmaintrue

\theoremstyle{plain}

\newcommand{\I}{\mathcal{I}}
\newcommand{\TSP}{\mathrm{TSP}}
\newcommand{\tsp}{\mathit{tsp}}
\newcommand{\MST}{\mathrm{MST}}

\newcommand{\short}{\mathrm{short}}

\newcommand{\Reals}{{\mathbb{R}}}

\newcommand{\etal}{\emph{et al.}\xspace}

\let\eps\varepsilon
\let\leq\leqslant
\let\le\leqslant
\let\geq\geqslant
\let\ge\geqslant

\newcommand{\mypara}[1]{\vspace{.4\baselineskip} \noindent \textbf{\sffamily #1}}

\newcommand{\peyman}[1]{{\footnotesize{\color{red}Peyman: #1}}}

\title{On Cyclic Solutions to the Min-Max Latency Multi-Robot Patrolling Problem}

  \author{Peyman Afshani}{Department of Computer Science, Aarhus University, Denmark}{peyman@cs.au.dk}{}{}
  
  \author{Mark de Berg}{Department of Mathematics and Computer Science, TU Eindhoven, the Netherlands}{M.T.d.Berg@tue.nl}{https://orcid.org/0000-0001-5770-3784}{Supported by the  Dutch Research Council (NWO) through    Gravitation-grant NETWORKS-024.002.003.}

  \author{Kevin Buchin}{Department of Computer Science, TU Dortmund, Germany}{kevin.buchin@tu-dortmund.de}{https://orcid.org/0000-0002-3022-7877}{}

  \author{Jie Gao}{Department of Computer Science, Rutgers University; New Brunswick, NJ 08854, USA}{jg1555@rutgers.edu}{https://orcid.org/0000-0001-5083-6082}{This work is supported by NSF OAC-1939459, CCF-2118953 and CCF-1934924.}

  \author{Maarten L\"{o}ffler}{Department of Information and Computing Sciences, Utrecht University, the Netherlands}{m.loffler@uu.nl}{}{}

  \author{Amir Nayyeri}{School of Electrical Engineering and Computer Science, Oregon State University, OR 97330, USA}{nayyeria@eecs.oregonstate.edu}{}{}

  \author{Benjamin Raichel}{Department of Computer Science; University of Texas at Dallas; Richardson, TX 75080, USA}{benjamin.raichel@utdallas.edu}{}{Partially supported by NSF CAREER Award 1750780.}

  \author{Rik Sarkar}{School of Informatics, University of Edinburgh, Edinburgh, U.K.}{rsarkar@inf.ed.ac.uk}{}{}

  \author{Haotian Wang}{Department of Computer Science, Rutgers University; New Brunswick, NJ 08854, USA}{hw487@cs.rutgers.edu}{}{}

  \author{Hao-Tsung Yang}{School of Informatics, University of Edinburgh, Edinburgh, U.K.}{haotsungyang@gmail.com}{}{}

\authorrunning{Peyman Afshani \etal}
\Copyright{Peyman Afshani, Mark de Berg, Kevin Buchin, Jie Gao,  Maarten L\"{o}ffler, Amir Nayyeri, Benjamin Raichel, Rik Sarkar, Haotian Wang, Hao-Tsung Yang}

\ccsdesc[100]{Theory of computation $\rightarrow$ Randomness, geometry and discrete structures $\rightarrow$  Computational geometry}

\keywords{Approximation, Motion Planning, Scheduling}

\begin{document}

\maketitle

\begin{abstract}
    We consider the following surveillance problem: 
    Given a set $P$ of $n$ sites in a metric space and a set~$R$ of $k$ robots with the same maximum speed,
    compute a \emph{patrol schedule} of minimum latency for the robots.
    Here a patrol schedule specifies for each robot an infinite sequence of sites to visit (in the given order) and the latency $L$ of a schedule is the maximum latency of any site, where the latency of a site $s$ is the
    supremum of the lengths of the time intervals between consecutive visits to $s$.
    
    When $k=1$ the problem is equivalent to the travelling salesman problem (TSP) 
    and thus it is NP-hard. For $k\ge 2$ (which is the version we are interested in) the problem becomes 
    even more challenging; for example, it is not even clear if the decision version of the problem is decidable, in particular in the Euclidean case.

    We have two main results. 
    We consider \emph{cyclic solutions} in which the set of sites must be partitioned into
    $\ell$ groups, for some~$\ell \leq k$, and each group is assigned a subset of the robots
    that move along the travelling salesman tour of the group at equal distance
    from each other.
    Our first main result is that approximating the optimal latency of the class of cyclic
    solutions can be reduced to approximating the optimal travelling salesman tour
    on some input, with only a $1+\varepsilon$ factor loss in the approximation
    factor and an $O\left(\left(  k/\varepsilon  \right)^k\right)$ factor loss in
    the runtime, for any $\varepsilon >0$.
    Our second main result shows that an optimal cyclic solution is a
    $2(1-1/k)$-approximation of the overall optimal solution. Note that for $k=2$ this
    implies that an optimal cyclic solution is optimal overall. We conjecture that
    this is true for $k\geq 3$ as well.

    The results have a number of consequences. 
    For the Euclidean version of the problem, for instance, combining our results with known results on Euclidean TSP, yields a
    PTAS for approximating an optimal cyclic solution, 
    and it yields a $(2(1-1/k)+\eps)$-approximation  
    of the optimal unrestricted (not necessarily cyclic) solution. 
    If the conjecture mentioned above is true, then our algorithm is actually
    a PTAS for the general problem in the Euclidean setting. 
    Similar results can be obtained by combining our results with other known TSP algorithms in 
    non-Euclidean metrics. 
\end{abstract}


\section{Introduction}
We study the following problem, motivated by the problem of monitoring a fixed set of locations
using autonomous robots:
We are given a set $P= \{ s_1, \cdots, s_n\}$ of $n$ sites in a metric space as well as a set $R =
\{r_1, \cdots, r_k\}$ of $k$ robots.
We assume the robots have the same maximum speed, called the \emph{unit speed}, and 
their task is to repeatedly visit (i.e., survey) the sites such that the maximum time during 
which any site is left unmonitored is minimized.
More precisely, we wish to compute a \emph{patrol schedule};
that is, an infinite sequence of sites to visit for each robot,
of minimum \emph{latency}. Here the latency of a site~$s_i$
is the supremum of the length of the time intervals between consecutive visits of~$s_i$,
and the latency of the patrol schedule is the maximum latency over all the sites.
\medskip

\mypara{Related Work.} 
For $k=1$, the problem reduces to the Traveling Salesman Problem. To see this,
consider the time interval $[0,3L]$, where $L$ is the optimal latency, and 
observe that every site is visited at least twice by the robot in this time interval. 
Let $L'\leq L$ be the maximum length of time between two consecutive visits of a site.
Then there exists a site that is visited at times $t_0$ and $t_0 + L'$ and 
all other sites are visited at least once in the time interval $(t_0, t_0 + L')$.
Hence, if an optimal solution has latency~$L$, there is a TSP tour of length at most~$L$.
The converse is clearly true as well---by repeatedly traversing a TSP tour of
length~$L$ we obtain a patrol schedule of latency~$L$---and so the TSP problem
is equivalent to the patrol problem for a single robot.
Since TSP is NP-hard even in the Euclidean case~\cite{papadimitriou1977euclidean} 
we will focus on approximation algorithms.
There are efficient approximation algorithms for TSP and, hence, for the
patrolling problem for~$k=1$. In particular, there is a $(3/2)$-approximation 
for metric TSP~\cite{christofides1976worst} (which was slightly improved very recently~\cite{improvedTSP})
and a PTAS for Euclidean TSP~\cite{arora1998polynomial, mitchell1999guillotine}.
However, it seems difficult to generalize these solutions to the case $k\geq 2$, 
because it seems non-trivial to get a grip on the structure of optimal solutions
in this case. We will mention some of the major challenges shortly.

There has been a lot of work on such surveillance problems in the robotics 
community~\cite{Elmaliach:2008:RMF:1402383.1402397,6094844,liujointinfocom2017,yang2019patrol,6106761,6042503}.
Most previous work, however, focused on either practical settings or aspects of the
problems other than finding the best approximation factor. 
There are two papers that provide theoretical guarantees for the weighted version of the problem,
where sites of higher weight require more frequent patrols.
Alamdari~\etal~\cite{alamdari2014persistent} provided a $O(\log n)$-approximation algorithm 
for the weighted problem for $k=1$.
(Due to existence of weights, a TSP tour may no longer be optimal for $k=1$.) 
Afshani~\etal~\cite{wayf} studied the problem for $k\geq 1$ and they present
an $O(k^2\log \frac{w_{\text{max}}}{w_{\text{min}}})$-approximation algorithm,
where $w_{\text{max}}$ and $w_{\text{min}}$ are the maximum and the minimum weights of the sites.

\mypara{Related Problems.}
As already mentioned, the TSP problem can be viewed as a special case of the problem for unweighted sites and
for $k=1$.
Another related problem is the \emph{$k$-path cover} problem where we want to find
$k$ paths that cover the vertices of an edge-weighted graph such that the maximum length of the paths is minimized. 
This problem has a $4$-approximation algorithm~\cite{arkin2006approximations}. 
Another problem is the problem of covering all the sites with $k$ trees 
that minimize the maximum length of the 
trees; this problem is known as the \emph{min-max tree cover} problem and it has
constant-factor approximation algorithms~\cite{arkin2006approximations,khani2014improved} with
$8/3$ being the current record~\cite{xu2013approximation}. 
The \emph{$k$-cycle cover} problem is similar, except that we want to use
$k$ cycles (instead of paths or trees); again constant-factor approximation algorithms
are known, with $16/3$ being the current record~\cite{xu2013approximation}. 
If the goal is to minimize the sum of all cycle lengths, there is a
$2$-approximation for the metric setting and a PTAS in the Euclidean
setting~\cite{khachai2015polynomial,khachay2016polynomial}. 
Our problem is also related to (but different from) the \emph{vehicle routing problem} (VRP)~\cite{dantzig1959truck},
which asks for $k$ tours, starting from a given depot, that minimize the total
transportation cost under various constraints;
see the surveys by Golden~\etal~\cite{golden2008vehicle} or T\'oth and Vigo~\cite{toth2002vehicle}.

\mypara{Our results.} 
All covering problems mentioned above are obviously decidable.
The question of decidability for the patrolling problem seems non-trivial. However,
since patrol schedules are infinite sequences and thus it is not even clear how 
to guess a solution\footnote{If we assume that all distances are integers and we want to decide whether the latency is at most a given integer $\ell$, then we can guess a solution as laid out in Appendix~\ref{appendix:decide}. These assumptions, however, do not hold in the Euclidean case, even if the coordinates of sites are rational.}.
To tackle this issue, we consider the class of \emph{cyclic solutions}. In a cyclic solution
the set $P$ of sites is partitioned into $\ell \leq k$ subsets $P_1, \cdots, P_\ell$, 
and each subset~$P_i$ is assigned $k_i$ robots, where $\sum_{i=1}^{\ell}k_i =k$.
The $k_i$ robots are then distributed evenly along a TSP tour of $P_i$,
and they traverse the tour at maximum speed. Thus, the latency of the sites
in $P_i$ equals~$\|T_i\|/k_i$, where $\|T_i\|$ is the length of the TSP tour of~$P_i$.

The significance of this definition is that in Section~\ref{sec:cyclic-approx} we prove that
(in any metric space) the best cyclic solution is a $2(1-1/k)$-approximation of the
optimal solution in terms of maximum latency. 
We do this by transforming an optimal solution to a cyclic one, with only a $2(1-1/k)$ factor loss
in the approximation ratio.
This proof is highly non-trivial and involves a number of graph-theoretic arguments and carefully inspecting the
coordinated motion of the $k$ robots, cutting them up at proper locations, and
re-gluing the pieces together to form a cyclic solution.
In combination with this, in Section~\ref{sec:cyclic} 
we prove that, given a $\gamma$-approximation algorithm for TSP, for any fixed $k$ and $\varepsilon>0$,
we can obtain a $(1+\varepsilon)\gamma$-approximation of the best cyclic schedule in polynomial time. 
Therefore, in the Euclidean setting, we can use a known PTAS to obtain a $(1+\eps)$-approximation 
to the \emph{best cyclic} solution and in the general metric setting, we can use known 
approximation algorithms for TSP~\cite{improvedTSP} to get a
$1.5$-approximation to the \emph{best cyclic} solution.
Together with the results in Section~\ref{sec:cyclic-approx} these lead to a
$(2-2/k+\eps)$-approximation algorithm for the Euclidean case, and a
$(3-3/k$)-approximation for general metrics. 

We conjecture that the best cyclic solution is in fact the best overall solution.
If this is true, then our algorithm in Section~\ref{sec:cyclic} already gives a PTAS
in the Euclidean setting. 
Observe that a corollary of our result in Section~\ref{sec:cyclic-approx} is that the conjecture
holds for $k=2$. We remark that there is an easy proof showing the existence of a cyclic $2$-approximation solution (See Section~\ref{subsec:challenges}). Our new bound $2(1-1/k)$ is a significant improvement when $k$ is a small constant. For example, for $k=3$, we get that a cyclic $4/3$ approximate solution exists, and for $k=2$ --as mentioned above-- that there is a cyclic optimal solution. 

\ignore{
\peyman{Need to move this part.}
First, partition the minimum spanning tree of the set of sites into a number of
connected components by removing all edges longer than a carefully chosen
threshold. Then, for each $\ell$ with $1\leq \ell\leq k$, take all possible
ways to combine the connected components into $\ell$ clusters. This gives us a
large collection of all the possible different partitions into clusters. For
every such partition we use the given $\gamma$-approximation algorithm for TSP
to approximate the optimal TSP tour of each cluster in the partition, and then
compute how to distribute the $k$ robots to these $\ell$ TSP tours so as to
minimize the maximum ratio of tour length to the number of robots. 
}

\section{Challenges, Notation, and Problem Statement}
\subsection{Notation and Problem Statement}
        

Let $(P,d)$ be a metric space on a set $P$ of $n$ sites, where the distance between two 
sites~$s_i,s_j \in P$ is denoted by $d(s_i,s_j)$. 
Following Afshani~\etal~\cite{wayf}, we model the metric space in the following way. 
We take the undirected complete graph $G=(P,P\times P)$, and we view each 
edge $(s_i,s_j)\in P\times P$ as an interval (that is, a continuous 1-dimensional space)
of length~$d(s_i,s_j)$ in which the robot can travel. This transforms the discrete
matric space~$(P,d)$ into a continuous metric space~$C(P,d)$.  From now on, and with 
a slight abuse of terminology, when we talk about 
the metric space $(P,d)$ we actually mean the continuous metric space $C(P,d)$.

We allow the robots to ``stay'' on a site for any amount of time. This implies
it never helps if a robot moves slower than the maximum speed:
indeed, the robot may as well move at maximum speed towards the next site and stay 
a bit longer at that site.
Also, it does not help to have a robot start at time $t=0$ ``in the middle'' of an edge,
so we can assume all robots start at some sites at the beginning. 
A \emph{schedule} of a robot~$r_j$ is defined as a
continuous function $f_j:\Reals^{\geq 0}\rightarrow C(P,d)$, where $f_j(t)$ specifies the position of~$r_j$ at time~$t$. The unit-speed constraint implies that
a valid schedule must satisfy $d(f_j(t_1),f_j(t_2))\leq |t_1-t_2|$ for all~$t_1,t_2$. 
A \emph{schedule for the collection~$R$ of robots}, denoted by~$\sigma(R)$, is a collection of schedules~$f_j$, 
one for each robot $r_j\in R$. Note that we allow robots to be at the same location at the same time.


We say that a site $s_i\in P$ is \emph{visited} at time~$t$ 
if $f_j(t)=s_i$ for some robot~$r_j$. Given a schedule $\sigma(R)$, the \emph{latency} $L_i$ 
of a site $s_i$ is defined as follows.
\[
L_i = \sup_{0\leq t_1<t_2} \{ |t_2-t_1| : \mbox{$s_i$ is not visited during the time interval $(t_1,t_2)$} \}
\]
%
%
We only consider schedules where the latency of each site is finite. Clearly such
schedules exist; e.g., a robot can repeatedly traverse a TSP tour of the sites.
Given a metric space~$(P,d)$ and a collection $R$ of $k$ robots, 
the \emph{(multi-robot) patrol-scheduling problem} is to find a schedule~$\sigma(R)$ minimizing the latency 
$L := \max\limits_i L_i$, the maximum latency of any site.  


\medskip

  


%

\subsection{Challenges} \label{subsec:challenges}

The problem of scheduling multiple robots is quite challenging and involves several subtleties, caused by the fact that patrol schedules are infinite sequences. For example, the time intervals between consecutive visits of any given site might
increase continuously, and so we have to define the latency of a site using the notion of supremum 
rather than maximum. Moreover,
for $k>1$, it is not even clear if the problem is decidable:
    Given a set of $n$ points in the
    Euclidean plane, an integer $k>1$, and a value $L$, is it decidable
    if there exists a patrol schedule for the $k$ robots
    such that the maximum latency is bounded by $L$?
As already mentioned, a corollary of our results in Section~\ref{sec:cyclic-approx} is that 
for $k=2$ there exists an optimal cyclic solution and thus for $k=2$ the answer to the above question is yes. 

A severe challenge is that, since patrol schedules are infinite sequences,
it is difficult to rule out chaotic solutions where the robots  visit the
sites in a way that avoids any sort of repeated pattern.
\begin{figure}[ht]
    \centering
    \includegraphics[scale=0.7]{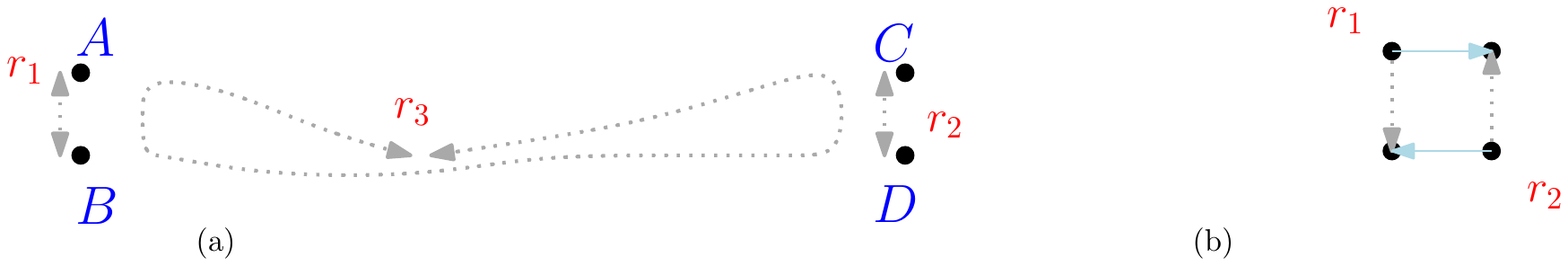}
    \caption{(a) Four points $A$, $B$, $C$, and $D$ form a short and wide rectangle. Robots $r_1, r_2$, and $r_3$ can have
    infinite ``unpredictable'' optimal patrol schedules.
    (b) Robots $r_1$ and $r_2$ can move to the other diagonal in two different ways. 
    }
    \label{fig:2groups}
\end{figure}
Indeed, optimal solutions can behave so chaotically 
that they require an infinite sequence of bits to describe. 
For instance, consider the left situation in Figure~\ref{fig:2groups}, 
where we have three robots and four points $A,B,C,D$
that are the vertices of a thin rectangle. To obtain the optimal latency, 
it suffices that $r_1$ moves back and forth between $A$ and $B$, and $r_2$
moves back and forth between $C$ and $D$. Since $r_3$ cannot be used to
decrease the latency---it will take $r_3$ too much time to go from
$A,B$ to $C,D$---it can behave as chaotically as it wants, thus causing
the description of the patrol schedule to be arbitrarily complicated.
This is even possible using only two robots: 
Consider four sites that form a unit square and two robots placed 
on opposite corners of the square; see the right situation in Figure~\ref{fig:2groups}. 
An optimal schedule is then an infinite sequence of steps, where in each step
both robots move counterclockwise or both move clockwise. Such a schedule
need not be cyclic and, hence, may require an infinite sequence of bits to describe. 
Of course, in both cases we know optimal cyclic solutions exist,
and such solutions can be described using finitely many bits.
We conjecture that this should be true in general:

\begin{conjecture}
\label{conj:opt_cyclic}
For the $k$-robot patrolling problem with min-max latency, there is a cyclic solution that is optimal. 
\end{conjecture}

It is easy to see that there exists a cyclic solution that is a $2$-approximation: 
take an optimal schedule with latency $L$, and at time $L$ move the robots back to 
their respective starting positions at time~$0$, and repeat. The challenge lies in getting an approximation factor smaller than~2, which we achieve in Section~\ref{sec:cyclic-approx} where we show that there is a cyclic solution 
that is a $2(1-1/k)$ approximation. 

\section{Turning an Optimal Solution into a Cyclic Solution}
\ifinmain
    \label{sec:cyclic-approx}
\else
    \label{app:cyclic-approx}
\fi

The main goal of this section is to prove the following theorem.
\ifinmain
    \begin{restatable}{theorem}{thmcyclic}\label{thm:opt-cyclic}
    Let $L$ be the latency in an optimal solution to the $k$-robot patrol-scheduling problem
    in a metric space~$(P,d)$. There is a cyclic solution with latency at 
    most~$2(1-1/k) L$.
    \end{restatable}
\else
    \thmcyclic*
\fi
%
We prove the theorem by considering an optimal (potentially ``chaotic'') solution and 
turning it into a cyclic solution. 
This is done by first identifying a certain set of ``bottleneck'' sites within a time interval
of length $L$, then cutting the schedules into smaller pieces, 
and then gluing them together to obtain the final cyclic solution. 
This will require some graph-theoretic tools (Appendix~\ref{app:graph}) as well as several new ideas (Appendices~\ref{app:sweep}--\ref{app:cyclic}).

\ifinmain

    Below we sketch the main ideas of the proof; 
    the full proof can be found in Appendix~\ref{app:cyclic-approx}.
    We also require some graph theoretic arguments that are presented in 
    Appendix~\ref{app:graph}.

\else

    \subsection{Graph Theoretic Preliminaries}\label{app:graph}
    Here, we use the term graph to include multigraphs 
(i.e., those with multiple edges between the same vertices and with loops).
A simple graph is one that does not have multiple edges or loops. 
A connected graph is Eulerian if all of its vertices have even degree. 
One of the earliest results in
graph theory is that an Eulerian graph has an Euler tour, i.e., a closed walk that visits every edge 
exactly once. 
Given a graph $G$, Eulerizing $G$ is the problem of duplicating the minimum number of edges of $G$ until the
resulting graph is Eulerian. 
We will use the following lemma.
\begin{lemma}\label{lem:eulerizing}
    Let $G$ be a connected graph with $E$ edges.
    \begin{enumerate}
        \item $G$ can always be Eulerized by duplicating all $E$ edges.
        \item If $G$ has exactly one vertex of degree 1, then we can Eulerize $G$ by duplicating at most  $E-1$ edges.
        \item If $G$ has no vertex of degree 1, then we can Eulerize $G$ by duplicating at most $E-2$ edges.
    \end{enumerate}
\end{lemma}
\begin{proof}
    The first claim is trivial.
    For the other two parts, let $T$ be a spanning tree of $G$.
    To make $G$ Eulerian, we need to make all the degrees even.
    Let $O$ be the number of vertices with odd degree in $T$.
    First observe that $O$ is an even number.
    Pair the vertices of odd degree into $O/2$ pair $v_i, u_i$, $1 \le i \le O/2$ and consider the
    path $\pi_i$ in $T$ that connects $v_i$ to $u_i$.
    Consider a pairing that minimizes the total length of the paths $\pi_i$.

    \begin{figure}[h]
        \centering
        \includegraphics[scale=0.5]{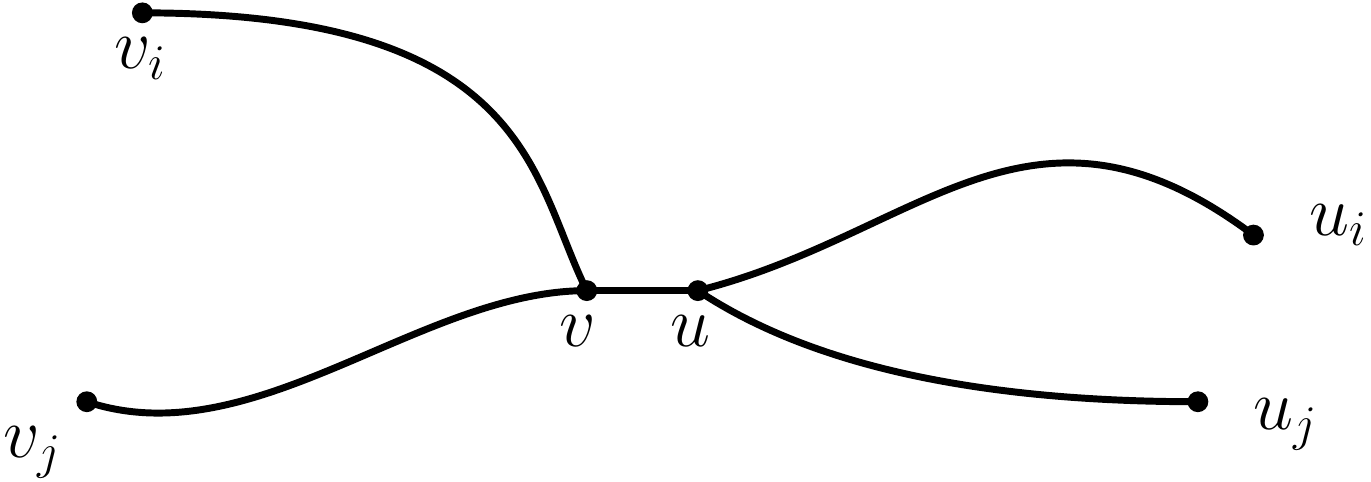}
        \caption{Two paths cannot share an edge, if the total path length is minimized.}
        \label{fig:dist}
    \end{figure}
    Observe that the paths 
    $\pi_1, \cdots, \pi_{O/2}$ will be edge disjoint. To see this, assume
    $\pi_i$ (connecting $v_i$ to $u_i$) and $\pi_j$ (connecting $v_j$ to $u_j$) share an edge $e= (v,u)$. 
    W.l.o.g, we can assume that on the path $\pi_i$ (resp. $\pi_j$) $v$ is closer to $v_i$ (resp. $v_j$) than $u$.
    See Figure~\ref{fig:dist}.
    Let $\pi_i(v_i,v)$ denote the subpath of $\pi_i$ from $v_i$ to $v$, and define $\pi_i(u,u_i)$, $\pi_j(v,v_j)$,
    and $\pi_j(u_j,u)$ similarly. Then the total length of the paths
    $\pi_i(v_i,v)\circ \pi_j(v,v_j)$ and $\pi_j(u_j,u)\circ \pi_i(u,u_i)$
    is one smaller than the total length of $\pi_i$ and $\pi_j$, which contradicts the choice of the pairing.
    Thus, the pairing that minimizes the total length of the paths consists of edge-disjoint paths. 
    
    After considering such a pairing, we simply duplicate all the edges on the paths $\pi_i$.
    Note that each edge in tree $T$ is duplicated at most once. 
    Now, the second claim in the lemma easily follows: as $G$ has exactly one vertex of degree 1, it is not a
    tree which means $T$ has at most $E-1$ edges and thus we duplicate at most $E-1$ edges to Eulerize $G$.

    For the third claim, we need some further case analysis. 
    Observe that if $T$ has at most $E-2$ edges, then we are done.
    Thus, assume $T$ has $E-1$ edges, so $G$ is the union of $T$ and one edge~$e$.
    But by our assumptions, $G$ does not have a degree 1 vertex which means $T$ has at most
    two leaves; but this means that $T$ is a tree with at most 2 leaves and thus $T$ can only be a path
    with $E-1$ edges and furthermore, both of these leaves must be connected by $e$ in $G$.
    Consequently, this implies that $G$ is a cycle but in that case, 
    $G$ is already Eulerian and thus nothing needs to be duplicated in this case. 
\end{proof}

A 2-path is the graph of two adjacent edges.
A $K_3$ or a triangle is the clique of three vertices and a claw is the complete bipartite graph $K_{1,3}$.
For a graph $G$, the line graph of $G$, denoted by $\overline{G}$, is a graph whose
vertices are the edges of $G$ and two vertices in $\overline{G}$ are connected if and only if 
the two corresponding edges in $G$ are adjacent. 
Observe that a perfect matching $M$ in $\overline{G}$ corresponds to a partition of edges of $G$ into 2-paths
and vice versa. 
The following theorems are known about the line graphs.
\begin{theorem}\label{thm:evenline}
    If $G$ is connected and it has an even number of edges, then $\overline{G}$ has a perfect matching $M$.
    Consequently, $M$ yields a partition of edges of $G$ into  
    a number of 2-paths.
\end{theorem}

However, we need something a bit more general about line graphs.
\begin{theorem}\label{thm:line}
    Consider a connected graph $G$ with odd number of edges.
    For any vertex $v$ in $G$, we can find a partition of edges of $G$ into
    a number of 2-paths plus an edge adjacent to $v$.
\end{theorem}
\begin{proof}
    The proof basically follows from Sumner's proof~\cite{sumner74}. 
    Pick any vertex $v$ in $G$.
    We use induction, meaning, we assume that the claim holds for any connected subgraph of $G$ that
    includes the vertex $v$.
    Let $T$ be the BFS traversal of ${G}$ with root $v$. 
    If there is a vertex $u$ such that $u$ is adjacent to two edges $e_1$ and $e_2$ with
    $e_1, e_2 \not \in T$, then we can remove $e_1$ and  $e_2$ as a 2-path and the remaining graph will
    stay connected (through $T$).
    And thus our claim follows by induction. 
    So in the rest of this proof assume that every vertex is adjacent to at
    most one edge that is not in $T$.

    Let $u$ be a leaf in $T$ that is farthest away from $v$ and let $e_1$ be
    the edge in $T$ adjacent to $u$.
    We now consider a few cases:
    \begin{itemize}
      \item case (i). Assume $u$ is adjacent to another edge $e_2$. Here, by the above assumption, $u$ is not adjacent to any other edge (in $G$).
    Now, we remove $e_1$ and $e_2$ as a 2-path, leaving us with an isolated vertex $u$ while the rest of the graph
    is still connected (through $T$) and thus our claim follows by induction.
    Thus assume, case (i) does not hold, which means $u$ has degree 1 in $G$. Consider the parent $w$ of $u$. 
      \item case (ii). Assume $w$ is adjacent to an edge $e_2$ that is not in $T$.
        In this case, we can remove $e_1$ and $e_2$ as a 2-path. This leaves $u$ as an isolated vertex and the rest of the graph
        will still stay connected and thus the claim follows by induction.

      \item case (iii). Assume $w$ is not adjacent to any edge that is not in $T$. Let $e_1$ be the edge that connects
        $w$ to $u$. This case has three additional subcases. 
        \begin{itemize}
          \item case (iii)a. Assume $u$ is the only child of $w$ and $w=v$. In this case, we are done, since we have only one edge
            left which is adjacent to $v$. 
          \item case (iii)b. Assume $u$ is the only child of $w$ and $w \not =v$. In this case, $w$ must have a parent and thus let
            $e_2$ be the edge that connects $w$ to its parent. 
            We can remove $e_1$ and $e_2$ as a 2-path. This leaves both $w$ and $u$ as two isolated vertices but the rest of the graph
            still stays connected and thus our claim follows by induction. 
          \item case (iii)c. Assume $w$ has another child $u'$ and let $e_2$ be the edge that connects $w$ to $u'$. 
            Here, we can remove $e_1$ and $e_2$ as a 2-path and this leaves both $u$ and $u'$ as two isolated vertices but the rest of the graph
            still stays connected and thus our claim follows by induction. 
        \end{itemize}
    \end{itemize}
\end{proof}

We now prove the one graph theoretic lemma that we will later use.
\begin{lemma}\label{lem:maingraph}
    Let $G$ be a connected graph that contains an even cycle. Let $\overline{G}$ be the line graph of
    $G$.
    Then, either $\overline{G}$ has a perfect matching or its vertices can be partitioned into a matching
    and a triangle.
    Consequently, 
    we can partition the edges of $G$ into a number of 2-paths  or into a number of 2-paths and 
    a claw.
\end{lemma}
\begin{proof}
    Observe that we only need to consider the case when $G$ has an odd number of edges as the other case
    is already covered by Theorem~\ref{thm:evenline}.

    Let $C$ be an even cycle in $G$, and let
    $G_1, \cdots, G_m$ be the connected components obtained after removing the edges of $C$.
    Since $G$ has an odd number of edges and $C$ has an even number of edges, there is
    a component $G_i$ with an odd number of edges. Let $v$ be a vertex of $G_i$
    that lies on $C$; such a vertex exists because $G$ is connected, so every component
    resulting from the removal of the edges of $C$ must contain at least one vertex from $C$.
    By Theorem~\ref{thm:line}, we can partition the edges of $G_i$ into a number of 2-paths 
    and at most one isolated edge $e$ adjacent to~$v$. Let $e'$ and $e''$ be the edges of $C$
    incident to $v$. Note that $\{e,e',e''\}$ is a claw. Now consider the graph $G$ again
    (including the cycle~$C$), and remove the claw $\{e,e',e''\}$ and the
    2-paths comprising $G_i\setminus \{e\}$ from $G$.
    The resulting graph~$G^*$ consist of the path $C\setminus  \{e',e''\}$
    plus various components $G_j$ attached to this path, and so $G^*$ is connected.
    Moreover, $G^*$ has an even number of edges, since it was obtained by removing an odd number of edges from $G$.
    Hence, the edges from $G^*$  can be decomposed into a number of 2-paths by Theorem~\ref{thm:evenline},
    thus finishing the proof.
\end{proof}


\fi

\ifinmain
\subsection{Bidirectional sweep to find ``bottleneck'' sites.}
\else
\subsection{Bidirectional sweep to find ``bottleneck'' sites.}\label{app:sweep}
\fi
Consider an optimal patrol schedule with latency $L$, and consider a time interval~$\I:=[t_0, t_0 + L]$
for an arbitrary $t_0 > 2L$. 
By our assumptions, every site is visited at least once within this time interval. 
We assign a time interval $I_i\subseteq \I$ to every robot~$r_i$. 
Initially $I_i = \I$. 
To identify the important sites that are visited by the robots, using a process that
we will describe shortly, we will \textit{shrink} each $I_i$.
Shrinking is done by moving the left and right endpoints of $I_i$ ``inward'' at the same speed. 
This will be done  in multiple stages and at the end of each stage, an endpoint of some intervals
could become \textit{fixed}; a fixed endpoint does not move anymore during the following stages.
When both endpoints of $I_i$ are fixed, we have found the final shrunken interval for $r_i$. 
Initially, all the endpoints are \textit{unfixed}.
For an interval $I_i = [t_i, t'_i]$, shrinking $I_i$ by some value $\varepsilon \ge 0$ yields the interval
$[t_i + \varepsilon \varphi_1, t'_i - \varepsilon \varphi_2]$ where $\varphi_1$ (resp. $\varphi_2$) is $1$ if
the left (resp. right) endpoint of $I_i$ is unfixed, otherwise it is 0. 
Note that an interval $[x,y]$ with $x>y$ is considered \emph{empty} (i.e., an empty set) and thus
shrinking an interval by a large enough value  will yield an empty interval
(assuming at least one endpoint is unfixed).

\mypara{The invariant.}
We maintain the invariant that at the beginning of each stage of the shrinking process, 
all the sites are visited during the shrunken intervals, i.e., 
for every site $s \in P$, there exists a robot $r_i$, and a value $t \in I_i$
such that $f_i(t) = s$.
Observe that 
the invariant holds at the beginning of the first stage of the shrinking process.

\mypara{The shrinking process.}
Consider the $j$-th stage of the shrinking process. 
Let $\varepsilon_j \ge 0$ be the largest (supremum) number such that shrinking
all the intervals by $\varepsilon_j$ respects the invariant. 
If $\varepsilon_j$ is unbounded, 
then this is the last stage; 
every interval $I_i$ that has an unfixed endpoint is reduced to an \emph{empty interval} and we are
done with shrinking, meaning, the shrinking process has yielded some $k'\le k$ intervals with both
endpoints fixed, and $k-k'$ empty intervals. 
Otherwise, $\varepsilon_j$  is bounded and well-defined as
the invariant holds for $\varepsilon_j = 0$.
With a slight abuse of the notation, let $I_1, \cdots, I_k$ be the intervals shrunken by $\varepsilon_j$. 
\ifinmain See Figure~\ref{fig:shrink}(left).  \else See Figure~\ref{fig:shrinkapp} (left).  \fi

Since $\varepsilon_j$ is the largest value that respects our invariant, it follows
that there must be at least one interval $I_{i_j}$ and at least one of its endpoints $t_j$ such that
at time $t_j$,  the robot  $r_{i_j}$ visited the site $f_{i_j}(t_j)$  
and this site is not visited by any other robot in the interior of their time intervals.
Now this endpoint of $I_{i_j}$ is marked as fixed and we continue to the next stage. 

For a fixed endpoint $A$, let $\ell(A)$ be the distance of $A$ to the corresponding boundary of the unshrunk interval.
More precisely, if $A$ is a left endpoint then the position of $A$ on the time axis is
$t_0 + \ell(A)$, and if $A$ is a right endpoint then this position is  $t_0 + L - \ell(A)$.
With our notation, if $A$ was discovered at stage $j$, then $\ell(A) = \varepsilon_1 + \cdots + \varepsilon_j$.

\begin{figure}[h]
  \centering
  \includegraphics[scale=0.65]{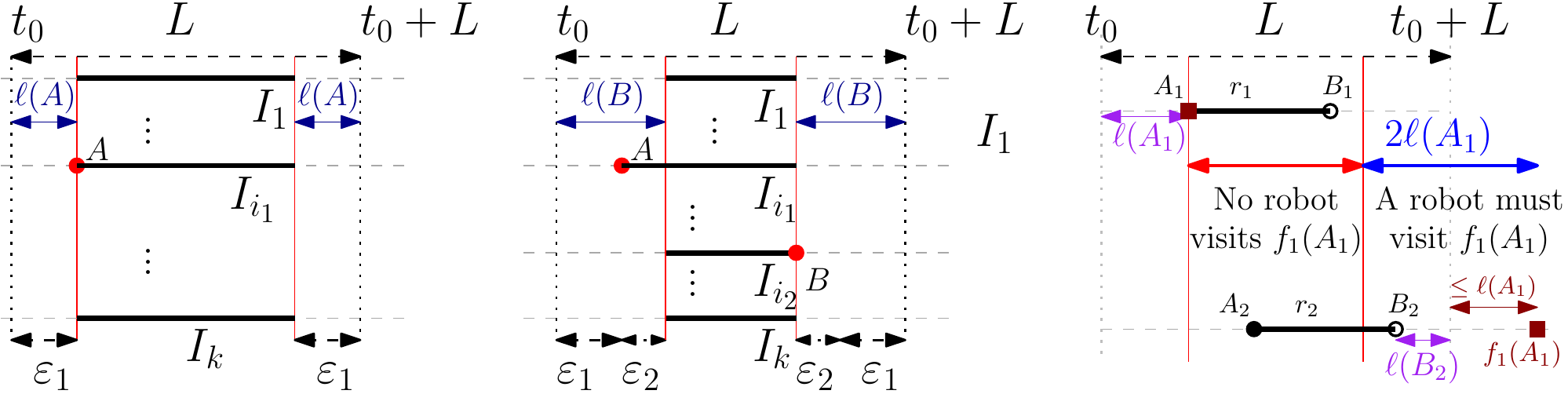}
  \caption{(left) $A$ is fixed at stage~1. (middle) $B$ is fixed at stage~2. 
  (right) By the property of the shrinking process, the site visited at $A_1$ is not visited by
  any robot within the red time interval but since the site has latency at most $L$, 
  it must be visited by some robot in the blue interval.
     }
\ifinmain
  \label{fig:shrink}
\else
  \label{fig:shrinkapp}
\fi
\end{figure}

\subsection{Patrol graph, shortcut graph, and bag graph}
\mypara{Shortcutting idea.}
\ifinmain Figure~\ref{fig:shrink} (right) \else Figure~\ref{fig:shrinkapp}(right) \fi
explains the crucial property of our shrinking process:
Robot $r_1$ visits the site $p=f_1(A_1)$ at time $A_1$ (which corresponds to the left endpoint of the interval~$I_1$) but to keep the latency of $p$ at most $L$, $p$
must be visited by another robot, 
say $r_2$, sometime in the interval $[t_0+L-\ell(A_1), t_0+L+\ell(A_1)]$,
shown in blue in the figure. 
For the moment, assume the right endpoint of the interval
of~$r_2$ is a fixed point $B_2$ and $r_2$ visits a site $p'=f_2(B_2)$ at time $B_2$.
This implies that the distance between $p$ and $p'$ is at most $\ell(A_1) + \ell(B_2)$.
Now observe that we can view this as a ``shortcut'' between endpoints $p$ and $p'$: for example,
$r_2$ can follow its own route from $A_2$ to $B_2$, then take the shortcut to $A_1$, and then follow
$r_1$'s route to $B_1$.
The extra cost of taking the shortcut, which is $\ell(A_1) + \ell(B_2)$,  can also be charged to
the two ``shrunken'' pieces of the two intervals (the purple intervals in the picture).
Our main challenge is to show that these shortcuts can be used to create a
cyclic solution with only a small increase in the latency.

To do that, we will define a number of graphs associated with the shrunken intervals. 
We define a \textit{patrol graph}  $\gP$, a
\textit{bag graph} $\gB$ and a \textit{shortcut graph} $\gS$.
The first two are multigraphs, whereas the shortcut graph is a simple graph. 
\ifinmain
For examples see Figure~\ref{fig:graphs} on page~\pageref{fig:graphs} and its
discussion on page~\pageref{par:example}.
\else
We will describe particular examples of these graphs below.
\fi

We start with the bag graph and the shortcut graph. 
We first shrink the intervals as described previously.
To define these graphs, consider $2k$ conceptual \textit{bags}, two for each
interval (including the empty intervals). More precisely, for
each interval we have one \emph{left bag} and one \emph{right bag}.
The bags are the vertices of the bag graph $\gB$ (Figure~\ref{fig:graphs}(b)).
The vertices of the shortcut  graph $\gS$ are the endpoints of the non-empty intervals.
To define the edges of the two graphs, we use \textit{placements}.  We will
present the details below but basically, every endpoint of a non-empty interval
will be placed in two bags,  one in some left bag  $\beta_1$ and another time
in some right bag $\beta_2$. 
After this placement, we add an edge in the bag graph between $\beta_1$ and $\beta_2$.  
Once all the endpoints have been placed, we add edges in the shortcut  graph
between every two endpoints that have been placed in the same bag (Figure~\ref{fig:graphs}(c)).

\ifinmain
    \begin{figure}[p]
      \centering
      \includegraphics[scale=0.7]{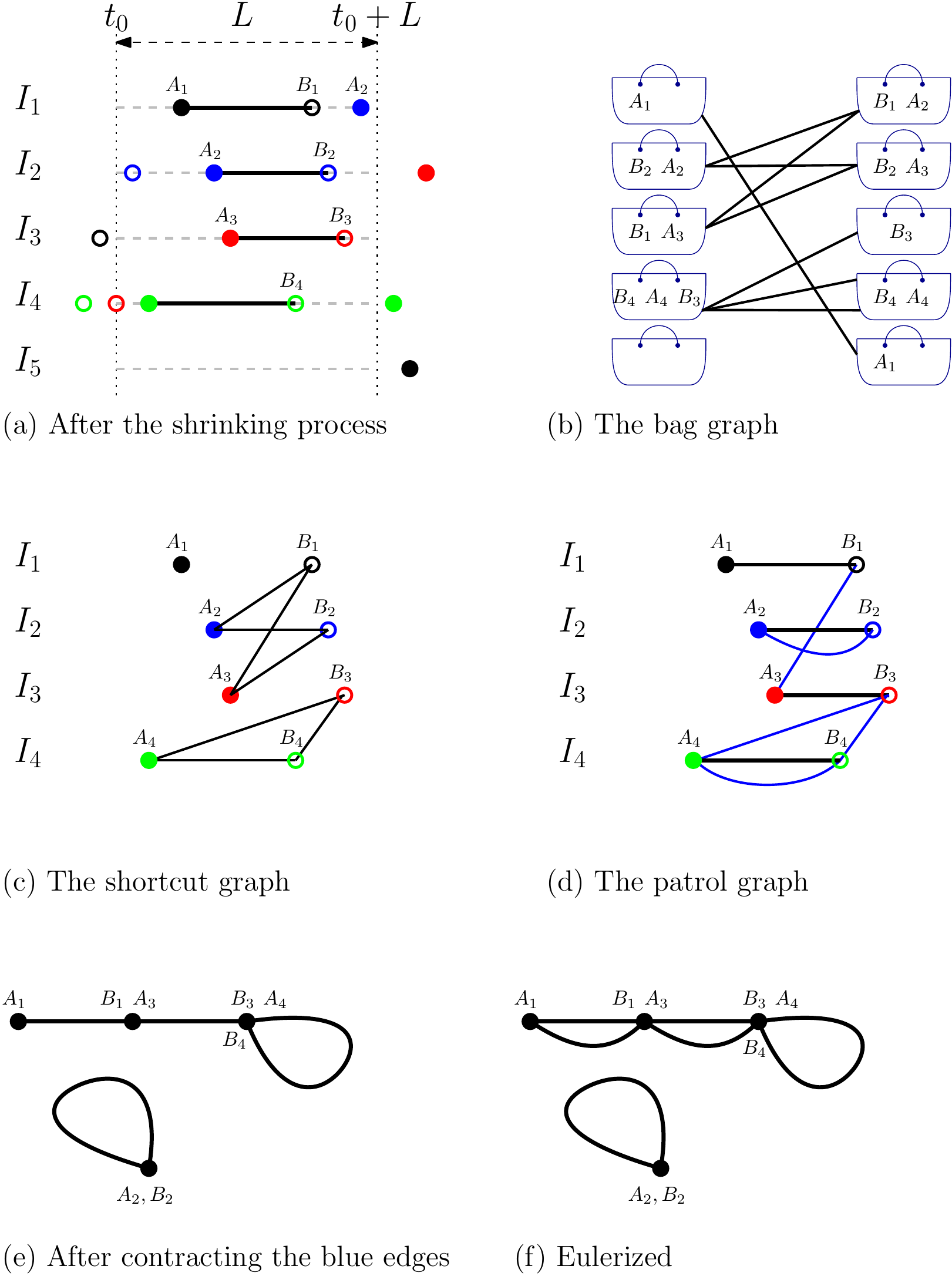}
      \caption{Example of bag, shortcut and patrol graph}
      \label{fig:graphs}
    \end{figure}
    \subparagraph{An example of a bag and shortcut graphs.}\label{par:example}
    An example is shown in Figure~\ref{fig:graphs}.
    In part (a), 
    we have four non-empty intervals $I_1 = [A_1,B_1],I_2 = [A_2,B_2],I_3 = [A_3,B_3], I_4 = [A_4,B_4]$
    and an empty interval $I_5$ (we will later explain the second appearance of
    each endpoint in this picture and for now the reader can ignore the
    ``floating'' endpoints).
    An example of a bag graph is shown in Figure~\ref{fig:graphs}(b):
    Every endpoint is placed twice (once in some left bag and one is some right bag). 
    E.g., $A_1$ is placed in the top-left bag and the
    bottom-right bag and thus the two bags are connected in the bag graph.
    Similarly, $B_1$ is placed in two
    bags, once at the top-right bag and the other time at the mid-left bag.  In
    part (c) of the figure, one can see the shortcut graph in which two endpoints
    are connected if and only if they are placed in the same bag. Also, this is a
    simple graph and despite the fact that $A_4$ and $B_4$ are placed together in
    two different bags, they are still connected once in the shortcut graph.  
\fi

Initially, all the bags are empty.
For every non-empty interval $I_1 = [A_1,B_1]$, 
we place the left endpoint of $I_1$ in its own left bag and the right endpoint of $I_1$ in its own right bag.
This is the first placement.
For the second placement, consider a non-empty interval $I_{i_1}$ and its left endpoint $A$.
The position of $A$ on the time interval is $t = t_0 + \ell(A)$. 
See 
\ifinmain Figure~\ref{fig:shrink}(right). \else Figure~\ref{fig:shrinkapp}(right). \fi
By our assumptions, the robot $r_{i_1}$ visits the site $p=f_{i_1}(t)$ at time $t$.
Consider the stage of our shrinking process when $A$ gets fixed.
For this to happen, the site $p$ cannot be visited by any robot
in the time interval $(t_0 +  \ell(A), t_0+L-\ell(A))$
(the red interval in 
\ifinmain Figure~\ref{fig:shrink}(right)), \else Figure~\ref{fig:shrinkapp}(right)), \fi
as otherwise, we could either shrink all the intervals
by an infinitesimal additional amount or some other endpoint would have been fixed.
On the other hand, this site has latency at most $L$,  so it must be visited by another robot
in the time interval $(t, t+L] = (t_0+\ell(A), t_0+\ell(A)+L]$.
This means that the robot $r_j$ that visits $p$ earliest in this interval must do so within
the time interval $[t_0 + L- \ell(A), t_0 + L + \ell(A)]$  (the blue interval in 
\ifinmain Figure~\ref{fig:shrink}(right)). \else Figure~\ref{fig:shrinkapp}(right)). \fi
Note that $r_j$ could be any of the robots, including $r_{i_1}$ itself.
We now place $A$ in the  right bag of $I_j$. 
\ifinmain
    A very similar strategy is applied to the right end point of $I_1$; for details see Section~\ref{app:cyclic-approx}
    where we also prove the following properties.
    \begin{restatable}{lemma}{oblist}\label{ob:list}
        The bag graph~$\gB$ and the shortcut  graph~$\gS$ have the following properties.
        \begin{enumerate}[(a)]
            \item $\gB$ is a bipartite graph.
            \item $\gS$ is isomorphic to the line graph of $\gB$
            \item Let $\gB'$  be a connected component of $\gB$.
                If none of the vertices of $\gB'$ belong to empty-intervals, then
                the number of vertices of $\gB'$ is equal to its number of edges.
            \item Let $\gS'$  be a connected component of $\gS$.
                If $\gS'$ has a vertex $v$ of degree one, then it must be the
                case that $v$ corresponds to an endpoint of a non-empty interval
                that has been placed (alone) in a bag of an empty interval. 
        \end{enumerate}
    \end{restatable}

    \mypara{The patrol graph.}
    The above lemma, combined with the graph theoretical tools that we outline in Appendix~\ref{app:graph}
    allows us to define the patrol graph $\gP$. 
    Here, we only give an outline, and for the full details see Appendix~\ref{app:cyclic-approx}.
    An example of a patrol graph is shown in Figure~\ref{fig:graphs}(d).
    Initially, the patrol graph, $\gP$, consists of $k'$ isolated black edges, one for each non-empty interval.
    Observe that both $\gP$ and the shortcut graph $\gS$ have the same vertex set (endpoints of the non-empty intervals).
    We add a subset of the edges of the shortcut graph to $\gP$.
    Let us consider an ``easy'' case to illustrate the main idea.

    \mypara{An easy case.}
    Assume $\gB$ is connected and that it has an even number of edges. 
    In this case, we can in fact prove that an optimal cyclic solution exists.
    Recall that $\gS$ is the line graph of $\gB$ and it is known (see Theorem~\ref{thm:evenline})
    that the line graph of a connected graph with even number of edges, has a perfect matching.
    Thus, we can find a perfect matching $M$ as a subset of edges of $\gS$. 
    Add  $M$ to $\gP$ as ``blue'' edges. Now, every vertex of $\gP$ is adjacent to
    a blue and a black edge and thus $\gP$ decomposes into a set of 
    ``bichromatic'' cycles, i.e., cycles with alternating black-blue edges.
    With a careful accounting argument, we can show that this indeed yields a cyclic solution without
    increasing the latency of any of the sites.
    We have already mentioned the main idea under the ``shortcutting idea'' paragraph, 
    at the beginning of the section.
    Specifically, we will use the following lemma.
    \begin{restatable}{lemma}{shortcut}\label{lem:shortcut}
      Consider two adjacent vertices $v$ and $w$ in the shortcut graph.
      This means that there are two non-empty intervals $I_1$ and $I_2$ such that $v$ corresponds to
      an endpoint $A$ of $I_1$ and $w$ corresponds to an end
      point $B$ of $I_2$ and $A$ and $B$ are placed in the same bag.
      Let $s_1$ be the site visited at $A$ during $I_1$ and $s_2$ be the site visited at $B$ on $I_2$.
      Then, we have $d(s_1, s_2) \le \ell(A) + \ell(B)$.
    \end{restatable}

    Black edges represent the routes of the robots, and blue edges are the shortcuts that connect one route
    to another. 
    So in this easy case, once the patrol graph has decomposed into bichromatic cycles, we turn 
    each cycle into one closed route (i.e., cycle) using the shortcuts. 
    All the robots that correspond to the black edges
    are placed evenly on this cycle.
    Since by our invariant all the sites are visited at some time on the black edges, it follows that the robots
    visit all the sites. 
    A careful accounting argument using the ``missing'' pieces $\ell(\cdot)$, then shows that the 
    latency does not increase at all.

    Unfortunately $\gB$ can have connected components with odd number of edges.
    Nonetheless, in all cases we can build a particular patrol graph, $\gP$, with the following
    properties.
    \begin{restatable}{lemma}{obpatrol}\label{ob:adjacency}
      The patrol graph $\gP$ consists of $k'$ pairwise non-adjacent black edges and a number of blue
      edges.  Any blue edge $(v,w)$ in $\gP$ corresponds to an edge in the shortcut graph $\gS$. 
      Furthermore, the set of blue edges
      can be decomposed into a matching and a number of triangles. In addition,
      any vertex of $\gP$ that is not adjacent to a blue edge can be charged to a
      bag of an empty interval. 
    \end{restatable}

    The idea covered in the above ``easy case'' works because it covers the black edges of $\gP$ with
    bichromatic edge-disjoint cycles and each cycle becomes a cyclic route.
    Unfortunately, in general $\gP$ might not have the structure that would allow us to do this.
    Here, we only outline the steps we need to overcome this:
    we consider each connected component, $\gP_i$, of $\gP$. 
    We first contract blue edges of $\gP_i$ to obtain a contracted patrol graph, $\gP_i^c$, 
    (Figure~\ref{fig:graphs}(e)) then
    we Eulerize it (Figure~\ref{fig:graphs}(f)), meaning, we duplicate a number of black edges such that the resulting graph is Eulerian.
    \begin{figure}[H]
      \centering
      \includegraphics[scale=0.7]{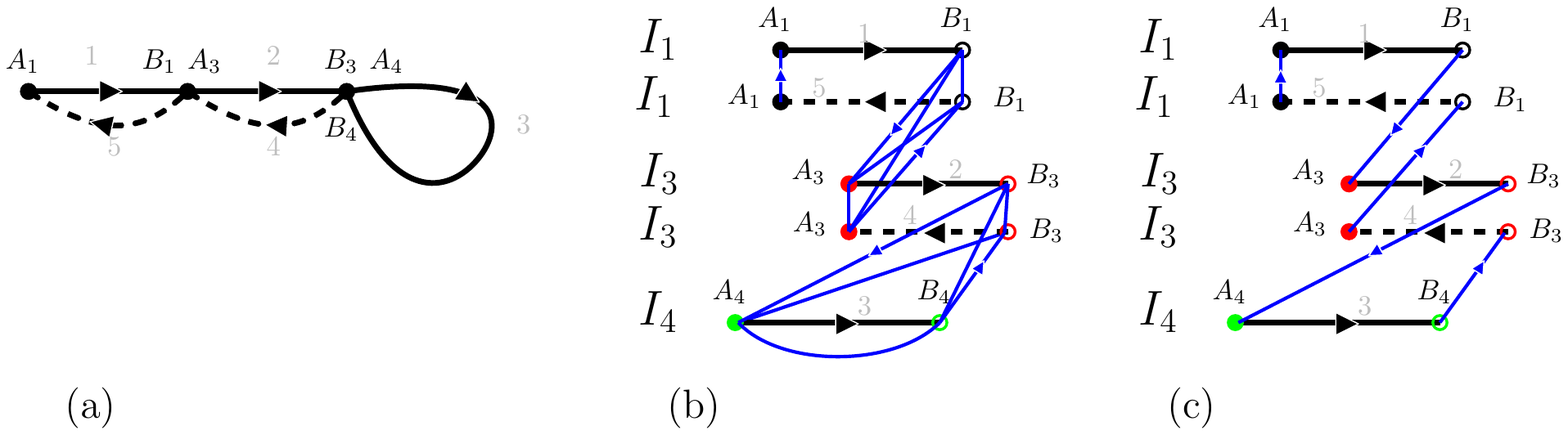}
      \caption{(a)  An Eularized contracted patrol graph (ECPG). Duplicated edges are drawn
        with dashed lines. The edges are directed and numbered according to an Euler tour.
      (b) The same ECPG after ``uncontracting'' the blue edges.
      The duplicated routes (i.e., the routes that will be traversed twice) are
      shown with dashed lines. 
      The corresponding Euler tour is marked and numbered.
      The shortcuts are taken in the correct direction.
      (c) The bichromatic cycle gives us a cyclic route. 
      It shows how the robots can travel along it. 
    }
      \label{fig:graphs-dup}
    \end{figure}
    This yields us an Eularized contracted patrol graph, $\gP_i^{{\scriptscriptstyle Ec}}$ (Figure~\ref{fig:graphs-dup}(a)).
    Next, we put the contracted blue edges back in $\gP_i^{{\scriptscriptstyle Ec}}$ which gives us the
    final patrol graph (Figure~\ref{fig:graphs-dup}(b)).
    In this final graph, we can show that we can cover the black edges using bichromatic edge disjoint
    cycles where each connected component of the final graph turns into one cycle (Figure~\ref{fig:graphs-dup}(c)); 
    this yields us a cyclic solution. 
    However, the duplicated black edges represent routes of robots that need to be traversed
    twice to obtain the cyclic solution. 
    This leads us to the final challenge: how to allocate the robots to the resulting cycles to minimize the
    latency.
    With some careful accounting and considering a few cases, we can show that this can be done in such a way
    that the resulting cyclic solution has latency at most $2L(1-1/k)$. We do this in Appendix~\ref{app:cyclic-approx}, proving Theorem~\ref{thm:opt-cyclic}.

\else

    A very similar strategy is applied to the right endpoint $B$ of $I_{i_1}$; 
    the location of $B$ on the time axis is $t' = t_0 + L - \ell(B)$, 
    the site $p'=f_{i_1}(t')$ cannot be visited by any robot within the
    interval $(t_0+\ell(B),t_0 + L - \ell(B))$ whereas it should be visited by some 
    robot in the region $[t'-L, t') = [t_0 - \ell(B), t')$ and thus
    it should visited by some robot in the time interval $[t_0 - \ell(B), t_0 + \ell(B)]$.
    Take $r_m$ to be the robot that visits $p'$ the latest in this time interval and then 
    the right endpoint of $I_{i_1}$ is placed in the left bag of $I_m$.
    This concludes the second placement of the endpoints of the non-empty intervals.

    Thus, for the sites $p$ and $p'$ defined above, we consider visits of them by the robots; intuitively, we consider once when are visited on the ``left'' side of the interval $\I$
    and another time on the ``right''.
    In Figure~\ref{fig:graphs}(a), the two visits of the sites are denoted by symbols of the same
    shape and color.
    E.g., $r_5$ is the first robot that visits the site
    $f_1(A_1)$ (colloquially speaking, the left endpoint of $I_1$) 
    during the interval $[t_0 + L- \ell(A_1), t_0 + L + \ell(A_1)]$.
    We now prove the following crucial lemma. 
    \oblist*
    \begin{proof}
    \begin{enumerate}[(a)]
    \item Every endpoint of a non-empty interval is placed
          twice, once in a left bag and once in a right bag. Hence,
          all edges of the bag graph are between left bags and 
          right bags. 
    \item Observe that by the definition of the two graphs, an edge $e$ in the bag graph between bags $\beta$ and $\beta'$ is formed because of the
        placement of an endpoint $p$ (of a non-empty interval) in the two bags. 
        Recall that $p$ is a vertex in the shortcut  graph $\gS$.
        We claim the shortcut  graph $\gS$ is isomorphic to the line graph of $\gB$ using the bijection that maps $e$ (a vertex in $\overline{\gB}$)
        to the vertex $p$ in $\gS$. 
        We now show that this is an isomorphism.

        Consider two  edges $e_1$ and $e_2$ in the bag graph in which $e_i$ is between bags $\beta_i$  and $\beta'_i$, and it is formed because of the placement of 
        an endpoint $p_i$ of a non-empty interval in the corresponding bags, for $i=1,2$.
        We consider two cases.
        First, assume, $e_1$ and $e_2$ are adjacent, and so w.l.o.g, assume $\beta_1=\beta_2$.
        For this to happen, $p_1$ and $p_2$ must be placed in the same bag $\beta_1 = \beta_2$.
        Thus, by the definition of the shortcut  graph, they are connected in $\gS$.
        The second case is when  $e_1$ and $e_2$ are not adjacent which implies all the bags $\beta_1, \beta'_1, \beta_2$, and $\beta'_2$ are
        distinct. 
        In this case observe that $p_1$ and $p_2$ are not placed in the same bag ever, since each of them is placed only twice, 
        and that no other bag can contain them.
        As a result, they are not connected with an edge in the shortcut  graph.
    \item Let $\beta_1, \dots, \beta_r$ be the vertices of $\gB'$.
        By assumption, none of the bags belong to an empty interval.
        However, recall that for a bag $\beta_i$, $1 \le i \le r$, the endpoint of its corresponding non-empty interval is placed in its own bag $\beta_i$, and then once more in some other bag
        $\beta'$.
        But since $\beta_1, \dots, \beta_r$ form a connected component, it must be the $\beta'$ should be one of the bags $\beta_1, \dots, \beta_r$, other than $\beta_i$. 
        As a result, we can charge the edge between $\beta'$ and $\beta_i$ to $\beta_i$ and it is clear that every bag will be charged exactly once, meaning,
        there would be exactly $r$ edges between the vertices $\beta_1, \dots, \beta_r$.
    \item Observe that if during the second placement, the left (resp. right) endpoint $p_1$ of a non-empty interval is placed in the right (resp. left) bag $\beta$ of a non-empty interval,
        then by our first placement, $\beta$ contains the right (resp. left) endpoint $p_2$ of some interval but crucially, $p_1 \not = p_2$ and thus $p_1$ will be connected to at least
        one other endpoint in the shortcut  graph.
        As a result, the only possible way for $p_1$ to end up with degree one is that $p_1$ is placed in a bag of an empty interval, and also no other endpoint is placed in that bag.
    \end{enumerate}
    \end{proof}

    \mypara{The patrol graph.}
    We are now ready to define the patrol graph $\gP$ (an example is shown in Figure~\ref{fig:graphs}(d)).
    Recall that $k'$ is the number of non-empty intervals. 
    Initially, $\gP$ consists of $k'$ isolated edges, one for each non-empty interval.
    Let us color these initial edges black. 
    We now define a natural bijection $J$ between the vertices of $\gP$ and the vertices of $\gS$:
    a vertex $v$ of $\gS$ is an endpoint of a non-empty interval and $J(v)$ maps it to the
    vertex in $\gP$ that represents the same endpoint. 
    In a process that we will describe shortly, we will add a subset of edges of $\gS$ to $\gP$
    and color them blue.
    To add an edge $e=(u,v)$ of $\gS$ to $\gP$, we simply add an edge between $J(u)$ and
    $J(v)$. 
    To describe the set of blue edges that we add, 
    consider a connected component $\gS'$ of $\gS$ and let $\gB'$ be the connected component
    of $\gB$ where $\gS' = \overline{\gB'}$.
    We now add a number of blue edges, using the following cases:
    \begin{itemize}
        \item (case i) $\gB'$ has an even number of edges. 
            Then by Theorem~\ref{thm:evenline} and Lemma~\ref{ob:list}(b), $\gS'$ has a perfect matching
            $M$. We add all edges of $M$ to $\gP$. 

            For example, in Figure~\ref{fig:graphs}, we add the perfect matching $\left\{A_2B_2 ,B_1A_3  \right\}$  in the central connected component to the
            patrol graph.
        \item (case ii) If case (i) does not hold, but none of the vertices of $\gB'$ belong to empty intervals, then
            by Lemma~\ref{ob:list}(c), the graph~$\gB'$ has at least one cycle $C$.
            However, by Lemma~\ref{ob:list}(a), the graph $\gB'$ is bipartite meaning, $C$ is an even cycle.
            Consequently, by Lemma~\ref{lem:maingraph}, we can decompose $\gB'$ into a number of 2-paths and a claw.
            The line graph of a claw is a triangle and thus it follows that the vertices of $\gS'$ can be decomposed into
            a matching $M$ and a triangle $\Delta$. 
            We add the edges of $M$ and $\Delta$ to $\gP$.

            For example, in Figure~\ref{fig:graphs} we have added the triangle formed by $A_4, B_4$, and $B_3$ to the patrol graph
            (here the matching part is empty).
        \item (case iii) If none of the cases (i) or (ii) hold, then at least one of the
            vertices of $\gB'$ belongs to an empty interval.
            In this case, we simply decompose $\gS'$ into a matching and an isolated vertex, since by Theorem~\ref{thm:line},
            $\gB'$ can be partitioned into a number of 2-paths and one edge.
            We add the edges in the matching
            to $\gP$ (in Figure~\ref{fig:graphs} we have $A_1$ as an isolated vertex).
            
    \end{itemize}

    \obpatrol*

    \subsection{Obtaining a Cyclic Solution}\label{app:cyclic}
    The main and the most crucial consequence  of our sweeping strategy is the following lemma. 
    \shortcut*
    \begin{proof}
      We have two cases: either both $A$ and $B$ are left or right endpoints, or one of them
      is a right endpoint while the other one is a left endpoint. 
      By symmetry, it suffices to consider the following two cases.
      \begin{itemize}
        \item $A$ and $B$ are both left endpoints. 
          For $A$ and $B$ to be placed in the same bag, they must be placed in the right bag of an interval $I_3$ (of a robot $r_3$) during the second placement of $A$ and $B$.
          Note that $I_3$ could be equal to both $I_2$ or $I_1$ but since $A$ and $B$ are both
          left endpoints, $I_1$ and $I_2$ are distinct intervals.
          As we argued during the second placement, $r_3$ visits $s_1$ in the
          time interval $[t_0+L-\ell(A),t_0+L+\ell(A)]$ and also visits $s_2$ in the time interval
          $[t_0+L-\ell(B),t_0+L+\ell(B)]$.
          As $r_3$ travels at unit speed, it thus follows that the distance between $s_1$ and $s_2$ cannot be larger than the maximum distance between the points of the two intervals
          $[t_0+L-\ell(A),t_0+\ell(A)+L]$, $[t_0+L-\ell(B),t_0+\ell(B)+L]$.
          Both of these intervals are centered at the point $t_0+L$, one has diameter $2\ell(A)$ and the other the diameter $2\ell(B)$ and thus the maximum distance between them is 
          $\ell(A) + \ell(B)$.
        \item $A$ is a left endpoint and $B$ is a right endpoint.
          Let $r_1$ and $r_2$ be the robots that correspond to $I_1$ and $I_2$.
          In this case, in order for $A$ and $B$ to be placed in the same bag, the second placement of $A$ should be placed in the
          same bag as the first placement of $B$ or vice versa. 
          W.l.o.g, assume the former, which implies $s_1$ must be visited by robot $r_2$
          during the time interval $[t_0+L-\ell(A), t_0 + L -\ell(A)]$; 
          however, observe that $s_2$ is visited by $r_2$ at time $t_0 + L -\ell(B)$ since this is the
          position of the endpoint $B$ on the time axis. 
          Once again, it follows that the distance between $A$ and $B$ is at most 
          the maximum distance between the point $t_0 + L -\ell(B)$ and the interval 
          $[t_0+L-\ell(A), t_0 + L + \ell(A)]$ which equals $\ell(A) + \ell(B)$.
      \end{itemize}
    \end{proof}

    Using the patrol graph defined previously, we can obtain an efficient cyclic solution. 
    Let $\gP_i$ be a connected component of $\gP$.
    We turn $\gP_i$ into one cycle. 
    Let $x$ be the number of black edges of $\gP_i$; it follows that $\gP_i$ has
    $2x$ vertices.
    Let $\gP_i^c$ be the graph obtained by contracting all the blue edges in $\gP_i$
    (see Figure~\ref{fig:graphs}(e)).
    We then use Lemma~\ref{lem:eulerizing} to duplicate a number of (black) edges of $\gP_i^c$ 
    such that the resulting graph, denoted by $\gP_i^{{\scriptscriptstyle Ec}}$, is Eulerian (see Figure~\ref{fig:graphs}(f)). 
    We call $\gP_i^{{\scriptscriptstyle Ec}}$ an Eularized contracted patrol graph (ECPG).
    We now look back at $\gP_i$.
    For every edge that is duplicated in $\gP_i^{{\scriptscriptstyle Ec}}$, we duplicate the edge as well as its
    two endpoints; this basically corresponds to adding another copy of the same interval to the picture. 
    In addition, we add a blue edge between the two copies of the same vertex as well as between 
    any two vertices that are connected in the shortcut graph. 
    See Figure~\ref{fig:graphs-dup} (b). 
    Let $\gP_i^{f}$  be the resulting graph. 
    The collection of all such graphs will form the \textit{final patrol graph} and
    thus $\gP_i^{f}$ is a connected component of the final patrol graph. 

    \begin{observation}\label{ob:final}
      Let $\gP_i^f$ be a connected component of the final graph. It has the following properties.
      $\gP_i^f$ consists of a number of pairwise non-adjacent black edges and a number of blue edges.
      Contracting all the blue edges of $\gP_i^f$ results in an  Eulerian graph.
      The blue edges can be partitioned into a number of vertex disjoint cliques.  
      For every blue edge $(v_1,v_2)$, the vertices $v_1$ and $v_2$ represent two
      endpoints $A_1$ and $A_2$ of two non-empty intervals $I_1$ and $I_2$ such that
      for the sites $s_i$ visited at $A_i$, $i=1,2$, we have
      $d(s_1,s_2) \le \ell(A_1) + \ell(A_2)$.
    \end{observation}
    \begin{proof}
      The first two claims are obvious due to our constructions and definitions.
      For the third claim, observe that by Lemma~\ref{ob:adjacency}, the blue
      edges of a patrol graph can be decomposed into (vertex disjoint) a number of $K_2$'s and $K_3$'s. 
      However, the duplication process also duplicates the blue edges, meaning, if a vertex of a blue clique $K_i$ is duplicated,
      it is still connected to all the other vertices of $K_i$ as well as it duplicate copy. 
      Thus, a blue clique $K_{i}$ can turn into a blue clique $K_{i'}$ where $i'\le 2i$.
      For the last property, if the edge $(v_i, v_2)$ is between the two copies of the same interval,
      then it follows that $I_i = I_2$ and thus $s_i = s_2$ and thus $d(s_i,s_2) = 0$.
      Otherwise, zero, one or both vertices of $v_i, v_2$ could be copies of another vertex.
      Assume $v'_i$ and $v'_2$ are the original vertices. By construction it holds that
      $(v'_i, v'_2)$ is a blue edge in a patrol graph and thus the last property holds
      by Lemma~\ref{lem:shortcut}.
    \end{proof}

    A \emph{bichromatic cycle} is a cycle of even size that is made of
    edges of alternating colors of black and blue.
    We can use the above observation to obtain the following result.
    \begin{lemma}\label{lem:onecycle}
      Let $\gP_i^f$ be a connected component of the final patrol graph with $x$
      black edges (the duplicates of the black edges are also counted).
      We can find a bichromatic cycle in $\gP_i^f$ that covers all the black edges.
      Consequently, it yields a cyclic route of length at most $xL$ that visits all the sites visited
      by the intervals that  correspond to the black edges of $\gP_i$.
    \end{lemma}
    \begin{proof}
      Let $\gP_i^{{\scriptscriptstyle Ec}}$ be the ECPG associated with $\gP_i^f$.
      Fix one Euler tour of $\gP_i^{{\scriptscriptstyle Ec}}$, i.e., an
      ordering $e_1, \cdots, e_x$ of the edges of $\gP_i^{{\scriptscriptstyle Ec}}$ 
      that forms a closed walk.  
      See Figure~\ref{fig:graphs-dup}(a). 
      The first part of the lemma is easy:
      for every $j$, $1 \le j \le x$, $e_j$ and $e_{j+1}$ (where $e_{x+1}=e_1$) are connected with a
      blue edge by Observation~\ref{ob:final} and thus the Euler tour in $\gP_i^{{\scriptscriptstyle Ec}}$
      corresponds to a bichromatic tour in the final patrol graph. 

      For the last part, 
      consider the step $j$ of the walk and let $I_j= [A_j,B_j]$ and  $I_{j+1}= [A_{j+1},B_{j+1}]$ be the intervals
      that correspond to the edges $e_j=(v_j, u_j)$ and $e_{j+1}=(v_{j+1}, u_{j+1})$ of
      $\gP_i^{{\scriptscriptstyle Ec}}$ where we assume $A_j$, $B_j$, $A_{j+1}$, and $B_{j+1}$
      correspond to $v_j$, $u_j$, $v_{j+1}$, and $u_{j+1}$ respectively.
      We make the convention that index $x+1$ refers to index 1 in the above definitions.
      Note that with this notation, $e_j$ and $e_{j+1}$ could be directed in either way in the Euler tour, i.e.,
      each edge can be directed from left to right, or right to left. 
      If $e_j$ is directed from $v_j$ to $u_j$ (i.e., from $A_j$ to $B_j$), 
      then define $\Start(e_j)=A_j$  and $\End(e_j)=B_j$ and if it is from $u_j$ to $v_j$ ($B_j$ to $A_j$), 
      then define $\Start(e_j)=B_j$  and $\End(e_j)=A_j$.
      Consider the same definitions for $e_{j+1}$.

      We now turn the Euler tour into a cyclic tour by considering an imaginary robot that travels along it.
      The robot starts at $\Start(e_1)$ and takes $j$ steps. 
      At step $j$, it travels the entire length of $e_j$, meaning, it traverse
      the route that corresponds to $e_j$, then it takes 
      the shortcut from $\End(e_j)$ to $\Start(e_{j+1})$ (to be more precise, it takes the shortcut from the
      sites that correspond to those endpoints) and then continues with step $j+1$. 
      Note that at the end of step $x$, it takes the shortcut from $\End(e_x)$ to $\Start(e_1)$. 
      Observe that for each edge $e_j$, we take one shortcut from $e_{j-1}$ to $e_{j}$ and one
      shortcut from $e_j$ to $e_{j+1}$.
      Next, observe that by Lemma~\ref{lem:shortcut}, we can charge the cost of the first shortcut to
      $\ell(\End(e_{j-1}))+\ell(\Start(e_{j}))$ and the cost of the second shortcut to 
      $\ell(\End(e_{j}))+\ell(\Start(e_{j+1}))$.
      The crucial point here is that each of $\ell(\Start(e_j))$ and $\ell(\End(e_j))$
      is used exactly once for each edge. 
      Finally, observe that $\ell(\Start(e_j))$ plus $\ell(\End(e_j))$ plus the length of 
      the interval corresponding to $e_j$ is exactly $L$.
      This shows that the imaginary robot travels a tour of length at most $xL$. 
    \end{proof}

        We are now almost done and we can finally describe our cyclic solution. 
        Bear in mind that in our scheme, we have $k'$ non-empty intervals that corresponds to tours that
        visit all the sites. 
        We say a robot is \textit{useful} if its interval is non-empty and \textit{useless} otherwise. 
    \ignore{
    If $\gP_i$ contains more than $k'$ vertices, then we call it the \textit{big} connected component.
    As $\gP$ has $2k'$ vertices, clearly there can be at most big connected component. 
    All the other connected components are \textit{small}.
    }
        Consider a connected component $P_i$ of the patrol graph (\textbf{not} the final patrol graph) 
        with $y$ black edges.
        If $y > k'/2$, then we call $P_i$ a \textit{big} connected component, otherwise, it is \textit{small}.
        Observe that by this definition, there can be at most one big connected component.
        Consider the graphs $\gP^c_i$, $\gP_i^{{\scriptscriptstyle Ec}}$, and $\gP_i^f$ associated with $\gP_i$. 
        We now consider a few cases. 
        \begin{itemize}
          \item (case 1) $\gP^c_i$ has at least two vertices $v_1$ and $v_2$ of degree 1. 
            Clearly, this means that $v_1$ and $v_2$ were not adjacent to any blue edges before
            contraction of blue edges, i.e., they have still degree 1 in $\gP_i$. 
            In this case, by Lemma~\ref{lem:eulerizing}, we can Eulerize $\gP^c_i$ by duplicating at most all the
            $y$ edges. As a result, $\gP_i^f$ has at most $2y$ edges and therefore by Lemma~\ref{lem:onecycle},
            we can turn it into a tour of length at most $2yL$.
            In this case, we can distribute all the $y$ useful robots corresponding to the $y$ non-empty intervals of
            $\gP_i$ and one useless robot. We have a useless robot available
            because by Lemma~\ref{ob:adjacency}, a vertex that is not adjacent to a blue edge can be charged to a
            bag of an empty interval and thus two such vertices can pay for a useless robot.
            Thus, using $y+1$ robots, we obtain a tour with latency
            \[
              \frac{2yL}{y+1} \le \frac{2(k-1)L}{k}
            \]
            where the inequality uses the fact that $y\le k-1$ (as there is at least one useless robot). 

          \item (case 2) $\gP^c_i$ has exactly one vertex $v_1$ of degree 1. 
            In this case, by Lemma~\ref{lem:eulerizing}, we can Eulerize $\gP^c_i$ by duplicating at most $y-1$
            edges. As a result, $\gP_i^f$ has at most $2y-1$ edges and therefore by Lemma~\ref{lem:onecycle},
            we can turn it into a tour of length at most $(2y-1)L$.
            But here we have two additional sub-cases. 
            If $P_i$ is not a big connected component, then we simply assign all the $y$ robots to the resulted
            Euler tour.
            Here, $P_i$ and $\gP^c_i$ have at most $k'/2$ edges, meaning, $y\le k'/2$.
            In addition, since $v_1$'s degree is 1,  there must be at least one empty bag---and, hence, at least one useless robot---by Lemma 5. Thus $k' \le k-1$. (Note that we are not using this useless robot here; we are still assigning only $y$ robots to the tour. We are only using the existence of the useless robot to conclude that $k’ \le k-1$.) Hence, the latency can be bounded as follows.
            \[ 
              \frac{(2y-1)L}{y} \le \frac{(2k'/2-1)L}{k'/2} = \frac{2(k'-1)L}{k'} < \frac{2(k-1)L}{k}.
            \]
            However, if $P_i$ is big, then we allocate one useless robot to the resulting tour.
            Note that here as $\gP^c_i$ only has one vertex of degree one, it can only be charged to one bag
            of an empty interval, i.e., ``half a useless robot''.
            But since we have at most one big component, we can afford to allocate a full useless robot to the resulting tour. 
            The latency here would be 
            \[ 
              \frac{(2y-1)L}{y+1} \le \frac{(2k'-1)L}{k'+1} \le \frac{(2k-3)L}{k} < \frac{2(k-1)L}{k}.
            \]

          \item (case 3) $\gP^c_i$ has no vertex of degree one. 
            In this case, by Lemma~\ref{lem:eulerizing}, we can Eulerize $\gP^c_i$ by duplicating at most $y-2$
            edges. As a result, $\gP_i^f$ has at most $2y-2$ edges and therefore by Lemma~\ref{lem:onecycle},
            we can turn it into a tour of length at most $(2y-2)L$.
            Thus, using $y$ robots, we obtain a tour with latency
            \[
              \frac{(2y-2)L}{y} \le \frac{2(k-1)L}{k}.
            \]
        \end{itemize}
        In all of the above cases, we obtain a cyclic tour of latency at most $2(k-1)L/k$ proving our
        main result. 

\section{On the decidability of the decision problem}\label{appendix:decide}

We note that the problem whether the latency is at most $\ell$ is decidable for the case that $\ell$ and the distances are all integers.\footnote{We thank a reviewer for pointing out this case.} To see this, consider the graph, in which we add dummy sites along the edges in $P \times P$ at unit distance.

We already know that we can assume that robots at time $t=0$ start at sites, and that at any time they either wait at a site or move at unit speed. Assume we have an optimal schedule $\sigma (R)$ with these properties and latency $L$, where $L\leq \ell$. We modify this schedule by rounding the times $t$ at which a robot leaves a site down the next integer $\lfloor t \rfloor$ (and then let it move with unit speed to the next site). 

First observe that in such a schedule robots always reach sites at integer times. Thus, inductively we can conclude that the robot that would have left a site at time $t$ already was at the site at time $\lfloor t \rfloor$  (and therefore can actually leave then).

Secondly the latency cannot increase to a value larger than $\ell$: The latency might increase by some value $x < 1$ and is an integer after adapting the schedule. If it would be larger than $\ell$, it would be at least $\ell+1$. But then the latency $L$ of the original schedule would have been $L \geq \ell + 1 - x > \ell$, a contradiction.

With these observations, we can now conclude that the problem is decidable. To describe a schedule, we only need to know for any time $t = 0, 1, \ldots$ at which site (original or dummy) it is.
Consider the schedule $\sigma(R)(t) = (f_1(t), \ldots, f_k(t))$ at such a time.
 If we have $N$ sites, there are only $N^k$ options for $\sigma(R)(t)$. This means by the pigeonhole principle that after $N^k +1$ time steps, we must have had a repeated schedule $\sigma(R)(t)$. At that point in time, we can make the schedule periodic without increasing the latency. Thus, by checking at most $\left(N^k\right)^{N^k +1}$ schedules we can decide the problem.

This however, leaves the case for points in $\mathbb{R}^2$ open, even if the points have integer coordinates.
\fi

\section{Cyclic Solutions}\label{sec:cyclic}
In this section we show how to approximate an optimal cyclic solution
to the patrol scheduling problem for $k$ robots in a metric space~$(P,d)$.  
We start with some notation and basic observations.

For a subset $Q\subseteq P$, let $\TSP(Q)$ denote an optimal TSP tour of $Q$ 
and let $\tsp(Q)$ denote its total length.  
Let $\MST(Q)$ denote a 
minimum spanning tree of $Q$. 
Now consider a 
partition $\Pi=\{P_1, \ldots, P_t\}$ of $P$, where each subset $P_i$
is assigned $k_i$ robots such that $\sum_{i=1}^t k_i =k$. A
\emph{cyclic solution} for this partition and distribution of robots
is defined as follows. For each $P_i$ there is a cycle $C_i$
such that the $k_i$ robots assigned to $P_i$ start 
evenly spaced along~$C_i$ and then traverse $C_i$ at maximum speed
in the same direction. Hence, the latency~$L$ of such a cyclic 
solution satisfies $L \geq \max_i (\tsp(P_i)/k_i)$, with equality 
if $C_i = \TSP(P_i)$ for all~$i$.

To prove the main theorem of this section we need several helper lemmas. 
Let $\Pi = (P_1, \ldots, P_t)$ be a partition of $P$ and let 
$E\subseteq P\times P$ be a set of edges. 
The \emph{coarsening of $\Pi$ with respect to $E$} is the partition~$\Pi'$ of $P$
given by the connected components of the graph $\left(\bigcup_i{\MST(P_i)}\right) \cup E$. 
\begin{lemma}
\label{lem:add_edges}
Let $S$ be a cyclic solution with partition $\Pi = (P_1, \ldots, P_t)$ and  latency $L$.  
Let $\Pi' = (P'_1, \ldots, P'_{t'})$ be the coarsening of $\Pi$ with respect to an edge set $E$ of total length~$\ell$.  
Then there is a cyclic solution $S'$ with partition $\Pi'$ and latency $L'$ such that
$L' \leq L + \ell$.
\end{lemma}
\begin{proof}
Let $C_1,\ldots,C_t$ be the cycles 
used in~$S$. Consider a subset $P'_i\in\Pi'$, and assume without 
loss of generality that $P_i'$ is the union of the subsets $P_1,\ldots,P_s$ from $\Pi$. 
Then there is a set $E_i\subseteq E$ of $k-1$ edges such that 
$\left(\bigcup_{j=1}^s C_j\right)\cup E_i$ is connected. Moreover, there is a 
cycle $C'_i$ covering all sites in $P'_i$ traversing the edges of each $C_j$
once and the edges of $E_i$ twice. Hence, 
\[
\|C'_i\| = \sum_{j=1}^s \|C_j\| + 2 \cdot \|E_i\|,
\]
where $\|\cdot\|$ denotes the total length of a set of edges.
Since the latency in $S$ is $L$, we know that
$\|C_j\| \leq k_j L$. Hence, using $\sum_{j=1}^s k_j\geq 2$ robots
for the cycle~$C'_i$, the latency for the sites in $P'_i$ is at most
\[
\frac{\|C'_i\|}{\sum_{j=1}^s k_j} 
= \frac{\sum_{j=1}^s \|C_j\|+ 2 \cdot \|E_i\|}{\sum_{j=1}^s k_j}
\leq \frac{\sum_{j=1}^s k_j L + 2 \cdot \|E_i\|}{\sum_{j=1}^s k_j}
\leq L + \|E_i\|
\leq L + \ell.
\]
Thus the latency for any subset $P'_i\in \Pi'$ is at most $L+\ell$.
\end{proof}

\begin{lemma}
\label{lem:structured_opt}
Let $L^*$ be the latency of an optimal cyclic solution. For any $\varepsilon>0$, 
there exists a cyclic solution with partition $\Pi=(P_1, \ldots, P_t)$ and
latency $L < (1+\varepsilon)L^*$ such that for any pair $i\neq j$
we have $d(P_i, P_j) > \varepsilon\cdot L^*/k$, where 
$d(P_i, P_j) := \min \{ d(x, y) : x\in P_i \mbox{ and } y \in P_j\}$.
\end{lemma}
\begin{proof}
Let $S^*$ be an optimal solution with partition $\Pi^* = (P^*_1, \ldots, P^*_q)$,
where $q \leq k$. Let $E_{\short}$ be the set of all edges of the complete graph 
of the metric space with length at most~$\varepsilon L^*/k$.  
Let $\Pi = (P_1, \ldots, P_t)$ be the partition obtained by coarsening $\Pi^*$
with respect to $E_{\short}$, and let $E^*\subseteq E_{\short}$ be a minimal subset 
such that coarsening $\Pi^*$ with $E^*$ gives the same partition $\Pi$.
Observe that as $q\leq k$, we have $|E^*|\leq k-1$.
Lemma~\ref{lem:add_edges} implies that there is a cyclic solution $S$ with 
partition $\Pi$ and latency at most
\[
L^* + |E^*|\cdot (\varepsilon L^*/k) < (1+\varepsilon)L^*.
\]
Moreover, since $\Pi$ is a coarsening of $\Pi^*$ with respect to $E_{\short}$, 
the pairwise distance between any two sets of $\Pi$ is larger than~$\varepsilon L^*/k$.
\end{proof}
\begin{lemma}
\label{lem:few_long_edges}
Suppose there is a cyclic solution of latency $L$ for a given metric space~$(P,d)$
and $k$ robots. Then $\MST(P)$ has fewer than $k(1+1/\alpha)$ edges of length more 
than $\alpha L$, for any $0<\alpha\leq 1$. 
\end{lemma}
\begin{proof}
Let $C_1,\ldots,C_q$ be the cycles in the given cyclic solution of latency~$L$, 
let $k_i$ denote the number of robots assigned to~$C_i$, and let $P_i\subset P$
be the sites in~$C_i$. Let $E$ be a subset of $q-1\leq k-1$ edges from $\MST(P)$
such that $\left(\bigcup_{i=1}^q C_i\right) \cup E$ is connected. 
Then  
\[
\sum_{i=1}^q \| C_i \|  > \| \MST(P)\| - \| E\| = \| MST(P)\setminus E\|
\]
Since $\| C_i \| \leq k_i L$, we have $\sum_{i=1}^q \| C_i \| \leq k L$.
Hence, $\| \MST(P)\setminus E\| < kL$, which implies that $\MST(P)\setminus E$
contains less than $k/\alpha$ edges of length more than $\alpha L$.
Including the edges in $E$, we thus know that $\MST(P)$ has less than $k(1+1/\alpha)$ 
edges of length more than $\alpha L$.
\end{proof}

\begin{theorem}
Suppose we have a $\gamma$-approximation algorithm for TSP in a metric 
space~$(P,d)$, with running time~$\tau_\gamma(n)$, and an algorithm for computing
an MST that runs in time~$T'(n)$. Then there is a 
$(1+\varepsilon)\gamma$-approximation algorithm for finding a minimum-latency
cyclic patrol schedule with $k$ robots that runs in
$T'(n)+ \left(O(k/\varepsilon)\right)^k\cdot \tau_\gamma(n)$ time.
\end{theorem}
\begin{proof}
Let $L^*$ be the latency in an optimal cyclic solution.
By Lemma~\ref{lem:structured_opt} there is a solution~$S$ with latency $(1+\varepsilon)L^*$ 
and partition $\Pi=\{P_1,\ldots,P_t\}$ such that $d(P_i,P_j)>\varepsilon L^* / k$ for all $i\neq j$.
Let $E$ be the set of edges of $\MST(P)$ with length more than $\varepsilon L^* /k$, 
and let  $T_1,\ldots,T_z$ be the forest obtained from $\MST(P)$ by removing~$E$. 
Let $V(T_j)$ denote the sites in $T_j$. For any $j$ we have $V(T_j)\subseteq P_i$ for some $i$. Otherwise, there would exist two sites $p,q\in V(T_j)$ that are neighbors in $T_j$ but stay in different sets in $\Pi$. 
This would lead to a contradiction: the former implies $d(p,q)\leq \varepsilon L^* /k$ 
while the later implies $d(p,q)> \varepsilon L^* /k$. 
Thus $\Pi$ is a coarsening of $\{V(T_1),\ldots,V(T_z)\}$ with respect to some subset of~$E$.

By Lemma~\ref{lem:few_long_edges}, the number of edges of $\MST(P)$ longer than $\varepsilon L^*/k$ 
is at most $k\left(1+\frac{k}{\varepsilon} \right)$. That is, the heaviest 
$k\left(1+\frac{k}{\varepsilon} \right)$ edges of $\MST(P)$ are a superset of
the set $E$ from above.  Thus we can find the partition $\Pi$ from above by first computing $\MST(P)$, 
removing the heaviest $k\left(1+\frac{k}{\varepsilon} \right)$ edges, and then trying all 
coarsenings determined by subsets of the removed edges.
Given a $\gamma$-approximation for TSP, below we argue how to get a $\gamma$-approximation to the 
optimal cyclic solution for a given partition $\Pi$. Running this subroutine for each of the above 
determined partitions and taking the best solution found will thus give latency at most~$(1+\varepsilon)\gamma L^*$.
\medskip

Observe that the optimal cyclic solution on a given partition $\Pi=\{P_1,\ldots,P_t\}$ 
uses cycles determined by $\TSP(P_i)$, and chooses $k_i$, the number of robots assigned to $P_i$,
so as to minimize $\max_i \tsp(P_i)/k_i$. Thus we can compute a $\gamma$-approximation of the optimal solution 
on $\Pi$ by first computing a $\gamma$-approximation to $\TSP(P_i)$ for all $i$, 
where $\overline{\tsp}(P_i)$ denotes its corresponding value, and then selecting $k_1',\ldots,k_t'$ 
so as to minimize $\max_i \overline{tsp}(P_i)/k_i'$. The latter step of determining the $k_i'$ 
can be done in $O(k\log k)$ time by initially assigning one robot to each $P_i$, and then iteratively 
assigning each next robot to whichever set of the partition currently has the largest ratio. 
The latency of the solution we find for $\Pi$ is thus
\[
\max_i \left\{ \frac{\overline{\tsp}(P_i)}{k_i'} \right\} \leq 
\max_i \left\{ \frac{\overline{\tsp}(P_i)}{k_i} \right\} \leq 
\max_i \left\{ \frac{\gamma \cdot \tsp(P_i)}{k_i} \right\} =
\gamma\cdot\mbox{[optimal cyclic latency for $\Pi$]},
\]
where the last inequality follows from the fact that
$\overline{\tsp}(P_i)\leq \gamma \cdot \tsp(P_i)$ for all $i$. 

It remains to bound the running time. For each partition~$\Pi$ we approximate $\TSP(P_i)$ 
for all~$i$, and then run an $O(k\log k)$ time algorithm to determine the robot assignment. 
Thus the time per partition is bounded by $\tau_\gamma(n)$, where $n$ is the total number of sites. Here we assume that 
$\tau_\gamma(n)=\Omega(n\log n)$ and that $\tau_\gamma(n)$ upper bounds the time 
for the initial $\MST(P)$ computation. 

The number of partitions we consider is determined by the number of subsets of size at most $k$ of the longest 
$k(1+k/\varepsilon)$ edges of $\MST(P)$, which is bounded by
\[
{{k(1+k/\varepsilon)}\choose k}\cdot 2^k,
\]
as the first term bounds the number of subsets of size exactly $k$, 
and for each subset the second term accounts for the number of ways in which we can pick at most $k$ edges from that subset.
We have the following standard upper bound on binomial coefficients.
\[
{N\choose K} \leq \left(\frac{N\cdot e}{K}\right)^K.
\]
Therefore, the total number of partitions we consider is at most
\[
{{k(1+k/\varepsilon)}\choose k}\cdot 2^k \leq  \left(\frac{k(1+k/\varepsilon)\cdot e}{k}\right)^k\cdot 2^k
= (2e(1+k/\varepsilon))^k
= \left(O(k/\varepsilon)\right)^k.
\]
Thus the total running time is $(O(k/\varepsilon))^k \cdot \tau_\gamma(n)$ as claimed.
\end{proof}
Recently, Karlin~\etal~\cite{improvedTSP} presented a $(3/2-\delta)$-approximation
algorithm for metric TSP, where $\delta>10^{-36}$ is a constant,
thus slightly improving the classic $(3/2)$-approximation by Christofides~\cite{christofides1976worst}.
Furthermore TSP in~$\Reals^d$ admits a PTAS~\cite{arora1998polynomial,mitchell1999guillotine}.
Thus we have the following.
\begin{corollary}\label{cor:euclidgeneral}
For any fixed $k$, there is polynomial-time $(3/2)$-approximation
algorithm for finding a minimum-latency cyclic patrol schedule with $k$ robots in arbitrary
metric spaces, and there is a PTAS in~$\Reals^d$ for any fixed constant~$d$.
\end{corollary}

Theorem~\ref{thm:opt-cyclic} in Section~\ref{sec:cyclic-approx} and Corollary \ref{cor:euclidgeneral} together 
imply the following.
    \begin{restatable}{theorem}{thmmain}\label{thm:main}
    For any fixed $k$ and $\varepsilon>0$, there is a polynomial-time 
    $(3(1-1/k)+\eps)$-approximation algorithm for the $k$-robot patrol-scheduling
    problem in arbitrary metric spaces, and a polynomial-time 
    $(2(1-1/k)+\eps)$-approximation algorithm in $\Reals^d$ (for fixed $d$).
    \end{restatable}





\section{Conclusion and Future Work}
This is the first paper that presents rigorous analysis and approximation algorithms 
for multi-robot patrol scheduling problem in general metric spaces. There
are several challenging open problems. The first and foremost is to prove or disprove
the conjecture that there is always a cyclic solution that is optimal overall.
Proving this conjecture will immediately provide a PTAS for the Euclidean multi-robot 
patrol-scheduling problem. It would also imply that the decision problem is decidable.
Another direction for future research is to extend the results to the weighted setting.
As has been shown for the 1-dimensional problem~\cite{wayf},
the weighted setting is considerably harder.

\bibliographystyle{plainurl}
\bibliography{ref_patrol} 

\newpage

\appendix

\inmainfalse



\end{document}